\documentclass[letterpaper, 10 pt, conference]{ieeeconf}  

\IEEEoverridecommandlockouts                              
\usepackage{amsmath}
\usepackage{amssymb}
\usepackage{amsfonts}
\usepackage[mathscr]{eucal}
\usepackage{graphicx} 

\usepackage{cite}

\newtheorem{theorem}{\it Theorem}

\newtheorem{corollary}{\it Corollary}
\newtheorem{proposition}{\it Proposition}

\title{\LARGE \bf
A Frequency-Domain Characterization of Optimal Error Covariance\\for the Kalman--Bucy Filter
}

\author{Song Fang, Hideaki Ishii, Jie Chen, and Karl Henrik Johansson
	\thanks{This research was supported by by the Knut and Alice Wallenberg Foundation, by the Swedish Strategic Research Foundation, by the Swedish Research Council, by JSPS under Grant-in-Aid for Scientific Research Grant
		No. 15H04020, by JST CREST under Grant No. JPMJCR15K3, by the Hong Kong RGC under the
		project CityU 111613, CityU 11260016, and by the City University of Hong Kong
		under Project 9380054.}
	\thanks{Song Fang and Karl Henrik Johansson are with the School of Electrical Engineering and Computer Science, Royal Institute of Technology, {\tt\small sonf@kth.se, kallej@kth.se}}%
	\thanks{Hideaki Ishii is with the Department of Computer Science, Tokyo Institute of Technology,  {\tt\small ishii@c.titech.ac.jp}}%
	\thanks{Jie Chen is with the Department of Electronic Engineering, City University of Hong Kong, {\tt\small jichen@cityu.edu.hk}}%
}

\begin{document}

\maketitle
\thispagestyle{empty}
\pagestyle{empty}

\begin{abstract}

In this paper, we discover that the trace of the division of the optimal output estimation error covariance over the noise covariance attained by the Kalman--Bucy filter can be explicitly expressed in terms of the plant dynamics and noise statistics in a frequency-domain integral characterization. Towards this end, we examine the algebraic Riccati equation associated with Kalman--Bucy filtering using analytic function theory and relate it to the Bode integral. Our approach features an alternative, frequency-domain framework for analyzing algebraic Riccati equations and reduces to various existing related results. 
\end{abstract}

\section{Introduction}

The Kalman--Bucy filter \cite{kalman1961new} is the optimal state estimator for linear continuous-time systems with white Gaussian noises. It is well known that the optimal state estimation error covariance attained by the Kalman--Bucy filter in steady state satisfies the algebraic Riccati equation \cite{anderson2012optimal, bittanti1991riccati, linearestimation}, which, however, does not provide an analytical expression of the error covariance. Over the years, various upper bounds and lower bounds have been obtained for the error covariance (see, e.g., \cite{kwon1996bounds} and the references therein). Meanwhile, nonrecursive algebraic solutions have been obtained as well (see, e.g., \cite{vaughan1969negative, bittanti1991riccati, crassidis2011optimal} and the references therein).

In this paper, we develop an alternative approach towards this problem from the viewpoint of the Bode integral \cite{Bod:45}. Bode integral was the first and arguably the best known result to analyze the fundamental limitations
of feedback control systems using analytic function theory \cite{francis1984h, Freu:85, boyd1985subharmonic, Mid:91, chen1995sensitivityc, chen2000logarithmic, stein2003respect}. In particular, the Bode integral is implied by the Poisson--Jensen formula \cite{goodwin2001control}. This line of research has been of continuing interest to the control community \cite{seron2012fundamental, Chen:14}, and was related to information theory in recent years as well \cite{zang2003nonlinear, Eli:04, Mar:07, Mar:08, Oka:08, Ish:09, Fang17TAC,Fang17Automatica, fang2017towards}; such results together with \cite{FangCDC17} inspired us to reexamine the Kalman--Bucy filter (in steady-state) using analytic function theory, e.g., the Poisson--Jensen formula, which will be seen to play an essential role in our analysis.




In particular, we first develop from the Poisson--Jensen formula a counterpart of Jensen's formula, based on which we obtain an analytical expression of the trace of the division of the optimal output estimation error covariance over the noise covariance by the Kalman--Bucy filter; this is given in terms of the plant dynamics, e.g., the unstable poles, and the noise statistics in a frequency-domain integral characterization. The analysis mainly concerns examining the algebraic Riccati equation associated with the Kalman--Bucy filter using analytic function theory. 

We compare our result with the existing related results in \cite{van1968detection, anderson1985use, braslavsky1999limiting, freudenberg2005control, zhang2012mmse}, and our integral characterization is seen to be the most generic and reduces to all of them. Specifically, the case when the plant is stable and the plant output is a scalar process was considered in \cite{van1968detection}. The authors of \cite{anderson1985use} then analyzed the case when the plant output is scalar but the plant is not necessarily stable. The following results in \cite{braslavsky1999limiting} considered the case when the process noise is relatively small in variance (and in the limit, zero) compared with the observation noise. Most recently, the case when the plant output is a vector process while the observation noise is with an identity covariance matrix was investigated in \cite{zhang2012mmse}. In addition, discussions on the dual problem in control can be found in \cite{anderson1985use, freudenberg2005control}.

The remainder of the paper is organized as follows. Section~II introduces the technical preliminaries. Section~III introduces a counterpart of Jensen's formula. In Section~IV, we examine the Kalman--Bucy filter using algebraic function theory and presents an integral characterization of the optimal output estimation error covariance. Relevant discussions and interpretations are also presented. Concluding remarks are given in Section~IV.


For analysis and discussions on the discrete-time Kalman filters, see our parallel work presented in \cite{FangACC18}.

\section{Preliminaries}

In this section, we introduce some relevant notions from stochastic processes, and briefly review the basic properties of the Kalman--Bucy filter. 

\subsection{Notations and Basic Concepts}

In this paper, we consider real-valued continuous zero-mean random variables and vectors, as well as continuous-time stochastic processes. We denote random variables and vectors using boldface letters. 
The logarithm $\ln$ is defined with base $\mathrm{e}$, and all functions are assumed to be measurable.  

A zero-mean stochastic process $\left\{ \mathbf{x} \left( t \right)\right\}, \mathbf{x} \left( t \right) \in \mathbb{R}^m$ is said to be (asymptotically) stationary \cite{Pap:02} if for any $t \geq 0$, the (asymptotic) correlation matrix 
\begin{flalign} 
&R_{\mathbf{x}}\left( \tau \right) =\lim_{t \to \infty} \mathrm{E}\left[ \mathbf{x} \left( t+\tau \right) \mathbf{x}^{T} \left( t \right) \right]\nonumber 
\end{flalign}  
exists. The (asymptotic) power spectrum of $\left\{ \mathbf{x} \left( t \right) \right\}, \mathbf{x} \left( t \right) \in \mathbb{R}^m$ is then defined as
\begin{flalign}
&\Phi_{\mathbf{x}}\left( \omega\right)
= \int_{-\infty}^{\infty} R_{\mathbf{x}} \left( \tau \right) \mathrm{e}^{-\mathrm{j}\omega \tau} \mathrm{d} \tau. \nonumber 
\end{flalign} 
It can be verified that $\Phi_{\mathbf{x}}\left( \omega\right)$ is positive semidefinite. In the scalar case, we denote $\Phi_{\mathbf{x}}\left( \omega\right)$ by $S_{\mathbf{x}}\left( \omega\right)$. For the inversion, 
\begin{flalign}
R_{\mathbf{x}} \left( \tau \right)
&= \lim_{t\to \infty} \mathrm{E}\left[ \mathbf{x} \left( t+\tau \right) \mathbf{x}^{T} \left( t \right) \right] \nonumber \\
&= \frac{1}{2\pi}\int_{-\infty}^{\infty} \Phi_{\mathbf{x}} \left( \omega\right) \mathrm{e}^{\mathrm{j}\omega \tau} \mathrm{d} \omega. \nonumber 
\end{flalign} 
Moreover, the (asymptotic) covariance matrix of $\left\{ \mathbf{x} \left( t \right)\right\}, \mathbf{x} \left( t \right) \in \mathbb{R}^m$ is given by
\begin{flalign}
&\Sigma_{\mathbf{x}}
= \lim_{t\to \infty} \mathrm{E}\left[ \mathbf{x} \left( t \right) \mathbf{x}^{T} \left( t \right) \right]
= R_{\mathbf{x}}\left( 0\right) 
= \frac{1}{2\pi}\int_{-\infty}^{\infty} \Phi_{\mathbf{x}}\left( \omega\right) \mathrm{d}  \omega. \nonumber 
\end{flalign}
In the scalar case, $\Sigma_{\mathbf{x}}$ reduces to (asymptotic) variance $\sigma_{\mathbf{x}}^2$. Moreover, if $\left\{ \mathbf{x} \left( t \right)\right\}, \mathbf{x} \left( t \right) \in \mathbb{R}^m$ is (asymptotically) white, then $\Phi_{\mathbf{x}}\left( \omega\right)$ is a constant $X$ for all $\omega$, and  
\begin{flalign}
&R_{\mathbf{x}} \left( \tau \right)
= \lim_{t\to \infty} \mathrm{E}\left[ \mathbf{x} \left( t+\tau \right) \mathbf{x}^{T} \left( t \right) \right]
= X \delta \left( \tau \right), \nonumber 
\end{flalign}
where $\delta \left( \cdot \right)$ is the Dirac delta function.

\subsection{The Kalman--Bucy Filter} \label{kalamnsection}

\begin{figure}
	\begin{center}
		\vspace{-3mm}
		\includegraphics [width=0.58\textwidth]{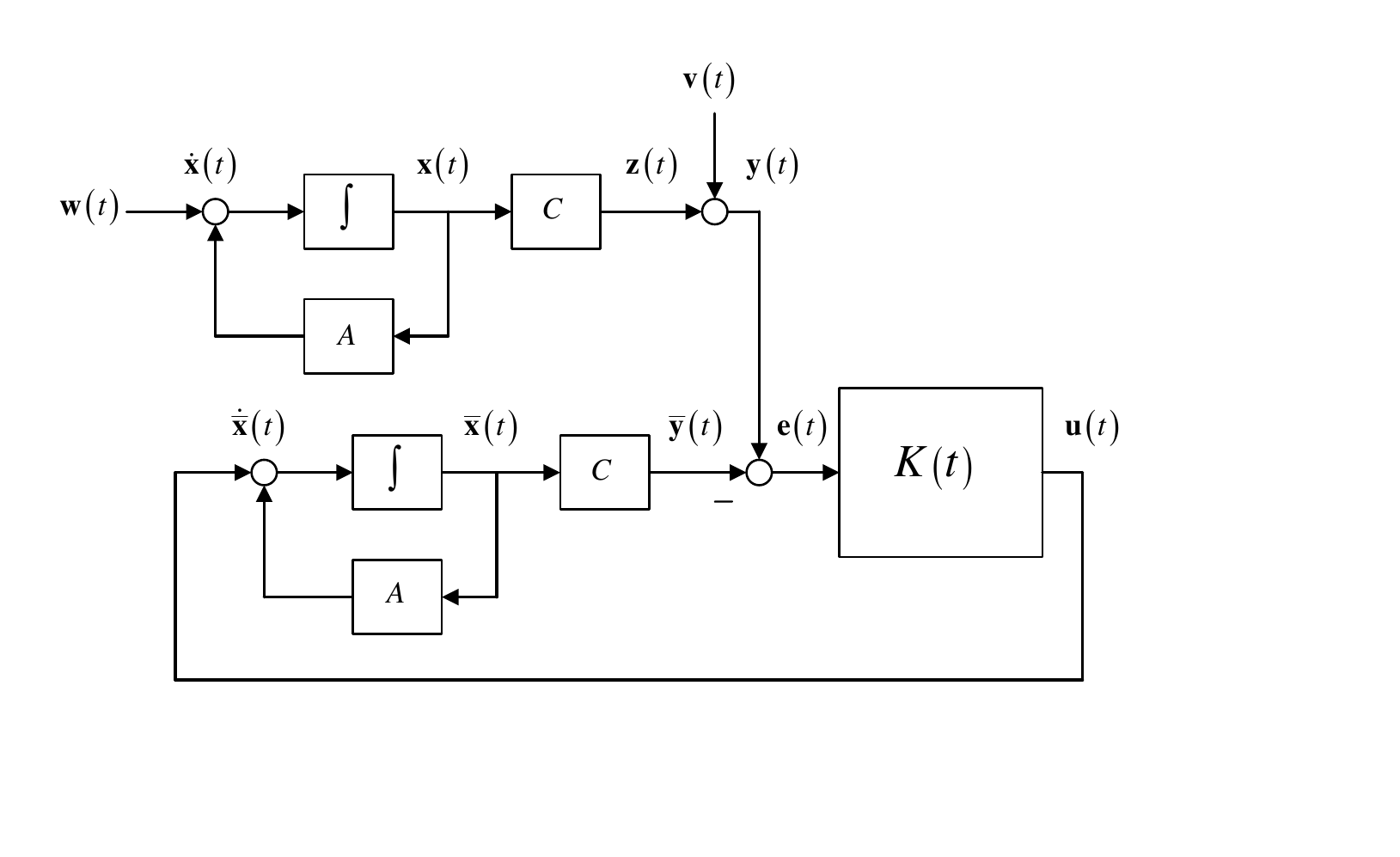}
		\vspace{-15mm}
		\caption{The Kalman--Bucy filter.}
		\label{Kalman}
	\end{center}
		\vspace{-6mm}
\end{figure}

Consider the continuous-time Kalman--Bucy filtering system \cite{linearestimation} depicted in Fig.~\ref{Kalman}, where the system is linear time-invariant (LTI) with state-space model given by
\begin{flalign}
&\left\{ \begin{array}{rcl}
\dot{\mathbf{x}} \left( t \right) & = & A \mathbf{x} \left( t \right) + \mathbf{w} \left( t \right),\\
\mathbf{y} \left( t \right) & = & C \mathbf{x} \left( t \right) + \mathbf{v} \left( t \right),
\end{array} \right. \nonumber 
\end{flalign}
where $\mathbf{x} \left( t \right) \in \mathbb{R}^{m}$ is the state to be estimated, $\mathbf{y} \left( t \right) \in \mathbb{R}^{l}$ is the system output, $\mathbf{w} \left( t \right) \in \mathbb{R}^{m}$ is the process noise, and $\mathbf{v} \left( t \right) \in \mathbb{R}^{l}$ is the measurement noise. The system matrices are $ {A} \in \mathbb{R}^{m \times m}$ and $ {C} \in \mathbb{R}^{l \times m}$, and $\left( A, C \right)$ is assumed to be detectable. Suppose that $\left\{ \mathbf{w} \left( t \right) \right\}$ and $\left\{ \mathbf{v} \left( t \right) \right\}$ are white Gaussian with covariance matrices 
\begin{flalign} 
&\mathrm{E} \left[ \mathbf{w} \left( t + \tau \right) \mathbf{w}^T \left( t \right) \right] = W \delta \left( \tau \right),~W \geq 0, \nonumber 
\end{flalign} 
and 
\begin{flalign} 
&\mathrm{E} \left[ \mathbf{v} \left( t + \tau \right) \mathbf{v}^T \left( t \right) \right] = V \delta \left( \tau \right),~V > 0, \nonumber 
\end{flalign} 
respectively, and that the initial state $\mathbf{x} \left( 0 \right) $ is Gaussian with covariance $\Sigma_{\mathbf{x} \left( 0 \right) }$ satisfying $0 < \det \Sigma_{\mathbf{x} \left( 0 \right) } < \infty$. Furthermore,  $\left\{ \mathbf{w} \left( t \right) \right\}$, $\left\{ \mathbf{v} \left( t \right) \right\}$, and $\mathbf{x} \left( 0 \right)$ are assumed to be mutually uncorrelated.  

The Kalman--Bucy filter is given by
\begin{flalign}
&\left\{ \begin{array}{rcl}
\dot{\overline{\mathbf{x}}} \left( t \right) & = & A \overline{\mathbf{x}} \left( t \right) + \mathbf{u} \left( t \right),\\
\overline{\mathbf{y}} \left( t \right) & = & C \overline{\mathbf{x}} \left( t \right), \\
\mathbf{e} \left( t \right) &=&\mathbf{y} \left( t \right) - \overline{\mathbf{y}} \left( t \right), \\
\mathbf{u} \left( t \right) &=& K \left( t \right) \mathbf{e} \left( t \right),
\end{array} \nonumber
\right. 
\end{flalign}
where $\overline{\mathbf{x}} \left( t \right) \in \mathbb{R}^m$, $\overline{\mathbf{y}} \left( t \right) \in \mathbb{R}^{l}$, $\mathbf{e} \left( t \right) \in \mathbb{R}^{l}$, and $\mathbf{\mathbf{u}} \left( t \right) \in \mathbb{R}^m$. Herein, $K \left( t \right) $ denotes the Kalman gain, which is given by
\begin{flalign} \label{gain}
& K \left( t \right) = P \left( t \right) C^T V^{-1}, 
\end{flalign}
where $P \left( t \right)$ denotes the state estimation error covariance matrix as
\begin{flalign} 
& P \left( t \right) = \mathrm{E} \left[ \left[ \mathbf{x} \left( t \right) - \overline{\mathbf{x}} \left( t \right)  \right] \left[\mathbf{x} \left( t \right)  -\overline{\mathbf{x}} \left( t \right)  \right]^T \right].\nonumber 
\end{flalign}
Herein, $P \left( t \right) $ is obtained using the Riccati equation
\begin{flalign} 
&\dot{P} \left( t \right) = AP \left( t \right) + P \left( t \right) A^{T} + W - P \left( t \right) C^{T}V^{{-1}} CP \left( t \right),  \nonumber
\end{flalign}
with $P \left( 0 \right) = \mathrm{E} \left[ \mathbf{x} \left( 0 \right) \mathbf{x}^T \left( 0 \right) \right].$

It is known that the Kalman--Bucy filtering system converges, i.e., the estimator is asymptotically stable, the state estimation error $\left\{ \mathbf{x} \left( t \right) - \overline{\mathbf{x}} \left( t \right) \right\}$ and the output estimation error $\left\{ \mathbf{e} \left( t \right) \right\}$ are asymptotically stationary, and $\left\{ \mathbf{e} \left( t \right) \right\}$ is asymptotically white, under the assumption that the system $\left( A, C \right)$ is detectable. Moreover, in steady state, the optimal state estimation error covariance matrix
\begin{flalign} 
&P =\lim_{t \to \infty} \mathrm{E} \left[ \left[ \mathbf{x} \left( t \right) - \overline{\mathbf{x}} \left( t \right)  \right] \left[\mathbf{x} \left( t \right)  -\overline{\mathbf{x}} \left( t \right)  \right]^T \right] \nonumber 
\end{flalign}
attained by the Kalman--Bucy filter satisfies the algebraic Riccati equation
\begin{flalign} \label{cARE}
& AP + PA^{T} + W - PC^{T}V^{{-1}} CP = 0, 
\end{flalign}
whereas the steady-state Kalman gain is given by 
\begin{flalign} \label{cAREK}
& K = P C^T V^{-1}.
\end{flalign} 
In addition, the optimal steady-state output estimation error covariance matrix is found to be
\begin{flalign} \label{cerrorcovariance}
\Sigma_{\mathbf{z} - \overline{\mathbf{y}}}
&= \lim_{t \to \infty} \mathrm{E}\left[ \left[ \mathbf{z} \left( t \right) - \overline{\mathbf{y}} \left( t \right) \right] \left[ \mathbf{z} \left( t \right) - \overline{\mathbf{y}} \left( t \right) \right]^{T} \right] \nonumber \\
&= C P C^{T},
\end{flalign}
where $\mathbf{z} \left( t \right) = C \mathbf{x} \left( t \right)$ denotes the true value of the system output (in comparison, $\mathbf{y} \left( t \right)$ denotes the measured value of the system output).
Moreover, when $l = 1 $, \eqref{cerrorcovariance} reduces to 
\begin{flalign}
&\sigma_{\mathbf{z} - \overline{\mathbf{y}}}^2
= \lim_{t\to \infty} \mathrm{E}\left[ \left[ \mathbf{z} \left( t \right) - \overline{\mathbf{y}} \left( t \right) \right]^2 \right]. \nonumber 
\end{flalign}
It is also worth mentioning that 
\begin{flalign}
& \lim_{t \to \infty} \mathrm{E}\left[ \mathbf{e} \left( t + \tau \right) \mathbf{e}^{T} \left( t \right) \right] 
= V \delta \left( \tau \right). \nonumber 
\end{flalign} 

\section{A Counterpart of Jensen's Formula}

In this section, we first examine the Poisson--Jensen formula (see, e.g., Appendix~C of \cite{goodwin2001control}) for a special class of analytic functions, and obtain a formula that will be essential for the rest of the paper.

\begin{proposition} \label{Jensen}
	Let \begin{flalign}
	& f \left( s \right) = \frac{p \left( s \right)}{q \left( s \right)} 
	\end{flalign} be a rational transfer function for which the numerator polynomial $p \left( s \right)$ and denominator polynomial $q \left( s \right)$ are both of order $m$. Suppose that all the poles of $f \left( s \right)$ are stable and that $f \left( s \right)$ has no zeros on the imaginary axis. In addition, suppose that 
	\begin{flalign}
	& \lim_{s \to \infty} f \left( s \right) = 1. 
	\end{flalign}
	Then, 
	\begin{multline} \label{cJensen}
	\frac{1}{2\pi}\int_{-\infty}^{\infty} \ln \left| f \left( \mathrm{j} \omega \right) \right| \mathrm{d} \omega \\
	=  \lim_{s \to \infty} \frac{1}{2} s \ln \left| f \left( s \right) \right| + \sum_{i=1}^{m} \max \left\{0, \Re \left[ \varphi_{i} \right] \right\},
	\end{multline}
	where $\varphi_{i}$ denote the zeros of $f \left( s \right)$.
\end{proposition} 

\begin{proof} See Appendix~A. \end{proof}

Note that in the discrete-time case (for the unit disk), Jensen's formula is a consequence of the more general Poisson--Jensen formula \cite{goodwin2001control}. Likewise in the continuous-time case (for the half plane), \eqref{cJensen} is also a consequence of the more general Poisson--Jensen formula \eqref{PJ}. As such, \eqref{cJensen} as well as the subsequent \eqref{cJensen2} may be viewed as the continuous-time (half-plane) counterparts of the discrete-time (unit-disk) Jensen's formula.

More generally, when not necessarily all the poles of $f \left( s \right)$ are stable, we can obtain the following result based on Proposition~\ref{Jensen}.

\begin{proposition} \label{Jensen2}
	Let \begin{flalign}
	& f \left( s \right) = \frac{p \left( s \right)}{q \left( s \right)} 
	\end{flalign} be a rational transfer function for which the numerator polynomial $p \left( s \right)$ and denominator polynomial $q \left( s \right)$ are both of order $m$. Suppose that $f \left( s \right)$ has no zeros on the imaginary axis and that 
	\begin{flalign}
	& \lim_{s \to \infty} f \left( s \right) = 1. 
	\end{flalign}
	Then, 
	\begin{flalign} \label{cJensen2}
	&\frac{1}{2\pi}\int_{-\infty}^{\infty} \ln \left| f \left( \mathrm{j} \omega \right) \right| \mathrm{d} \omega \nonumber \\
	&\ \ \ \  =  \lim_{s \to \infty} \frac{1}{2} s \ln \left| f \left( s \right) \right| + \sum_{i}^{} \max \left\{0, \Re \left[ \varphi_{i} \right] \right\}
	\nonumber \\
	&\ \ \ \ \ \ \ \  - \sum_{j}^{} \max \left\{0, \Re \left[ \eta_{j} \right] \right\}, 
	\end{flalign}
	where $\varphi_{i}$ denote the zeros of $f \left( s \right)$ and $\eta_{j}$ denote its poles.
\end{proposition}

\begin{proof} See Appendix~B. \end{proof}

\section{Optimal Error Covariance by the Kalman--Bucy Filter}

We now study the optimal output estimation error covariance matrix $\Sigma_{\mathbf{z} - \overline{\mathbf{y}}}$ by the continuous-time Kalman--Bucy filter, which is determined by the plant dynamics and the noise statistics in an integral characterization.

\begin{theorem} \label{cgeneral}
	The optimal output estimation error covariance matrix $\Sigma_{\mathbf{z} - \overline{\mathbf{y}}}$ by the Kalman--Bucy filter satisfies
	\begin{flalign} \label{cgeneral1}
	\mathrm{tr} \left( \Sigma_{\mathbf{z} - \overline{\mathbf{y}}} V^{-1} \right) 
	&= \mathrm{tr} \left( CPC^T V^{-1} \right) \nonumber \\
	&=  \frac{1}{2\pi}\int_{-\infty}^{\infty} \ln \left[ \frac{\det \Phi_{\mathbf{y}} \left( \omega \right)}{\det V} \right] \mathrm{d} \omega \nonumber \\
	&\ \ \ \  + 2 \sum_{i=1}^{m} \max \left\{0, \Re \left[ \lambda_{i} \left( A \right)\right] \right\}, &
	\end{flalign}
		where
		\begin{flalign} \label{cnotation}
		\Phi_{\mathbf{y}} \left( \omega \right) =  C \left( \mathrm{j} \omega I - A \right)^{-1} W \left( - \mathrm{j} \omega I - A \right)^{-T} C^T  + V,
		\end{flalign}
	and $\lambda_{i} \left( A \right)$ denote the eigenvalues of $A$.
\end{theorem}


\begin{proof} See Appendix~C. \end{proof} 

On the right-hand side of \eqref{cgeneral1}, the term 
\begin{flalign}
&\sum_{i=1}^{m} \max \left\{0, \Re \left[ \lambda_{i} \left( A \right)\right] \right\} 
\end{flalign}
quantifies the instability of the system \cite{qiu2010quantify, chen2013stabilization}. Moreover, when $A$ is stable, i.e., when all the eigenvalues $\lambda_{i} \left( A \right)$ of matrix $A$ satisfy $ \Re \left[ \lambda_{i} \left( A \right)\right] < 0$, $\left\{ \mathbf{y} \left( t \right) \right\}$ is stationary and thus
\begin{flalign}
&C \left( \mathrm{j} \omega I - A\right)^{-1} W \left(-\mathrm{j} \omega I - A\right)^{-T} C^T  + V 
\end{flalign}
denotes the power spectrum $\Phi_{\mathbf{y}} \left( \omega \right)$ of $\left\{ \mathbf{y} \left( t \right) \right\}$; in general, however, $A$ needs not to be stable, and hence the function
\begin{flalign} \label{cpopov}
&C \left( s I - A\right)^{-1} W \left( -s I - A\right)^{-T} C^T  + V 
\end{flalign}
is a Popov function (see, e.g., \cite{linearestimation}, for a detailed discussion). 

Note that an alternative proof of \eqref{cinvBode_dynamic} is given by \cite{mustafa1990minimum} as 
\begin{flalign} 
&\frac{1}{2\pi}\int_{-\infty}^{\infty} \ln \left| \det \left[ I + C \left( \mathrm{j} \omega I - A \right)^{-1} K \right]^{-1} \right|^2 \mathrm{d} \omega \nonumber \\
&\ \ \ \  =\frac{1}{2\pi}\int_{-\infty}^{\infty} \ln \left| \det \left[ I - C \left[ \mathrm{j} \omega I - \left(A - KC\right) \right] ^{-1} K \right] \right|^2 \mathrm{d} \omega \nonumber \\
&\ \ \ \  = - \mathrm{tr} \left( CK \right)+ 2 \sum_{i=1}^{m} \max \left\{0, \Re \left[ \lambda_{i} \left( A - KC + K C \right)\right] \right\} \nonumber \\
&\ \ \ \  = - \mathrm{tr} \left( C P C^{T} V^{-1} \right) + 2 \sum_{i=1}^{m} \max \left\{0, \Re \left[ \lambda_{i} \left( A \right)\right] \right\}. 
\end{flalign}
One may refer to Appendix~A of \cite{mustafa1990minimum} for further details on this. 

\subsection{Interpretation from the Viewpoint of Bode Integral} \label{BodeIntegral}

Equation \eqref{cinvBode_dynamic} in the proof of Theorem~\ref{cgeneral}, in particular,
\begin{flalign} \label{cinvBode_dynamic1}
&\frac{1}{2\pi}\int_{-\infty}^{\infty} \ln \left| \det \left[ I + C \left( \mathrm{j} \omega I - A \right)^{-1} K \right]^{-1} \right| \mathrm{d} \omega \nonumber \\
&\ \ \ \  =  - \frac{1}{2} \mathrm{tr} \left( C P C^{T} V^{-1} \right) + \sum_{i=1}^{m} \max \left\{0, \Re \left[ \lambda_{i} \left( A \right) \right] \right\}, 
\end{flalign}
can also be obtained as a consequence of the (continuous-time) Bode integral \cite{goodwin2001control}; it is interesting to discover that such a relation exists between these two top equations \cite{topequations}, namely algebraic Riccati equation and Bode integral, of the control field. Specifically, the estimator may be viewed as a feedback system, and thus the Bode integral can be obtained for its sensitivity from $\left\{ \mathbf{y} \left( t \right) \right\}$ to $\left\{ \mathbf{e} \left( t \right) \right\}$ (see Fig.~\ref{Kalman}). 
Indeed, if we denote 
\begin{flalign}
& L \left( s \right) = C \left( s I - A \right)^{-1} K, 
\end{flalign}
then \eqref{cinvBode_dynamic1} can be rewritten as 
\begin{multline} \label{cBode}
\frac{1}{2\pi}\int_{-\infty}^{\infty} \ln \left| \det \left[ I + L \left( \mathrm{j} \omega \right) \right]^{-1} \right| \mathrm{d} \omega \\
=   - \lim_{s \to \infty} \frac{1}{2} s \mathrm{tr} L \left( s \right) + \sum_{i=1}^{m} \max \left\{0, \Re \left[ \lambda_{i} \right] \right\},
\end{multline}
where $\lambda_{i} $ denote the poles of $L \left( s \right)$. Herein, we have used the fact that 
\begin{flalign}
& \sum_{i=1}^{m} \max \left\{0, \Re \left[ \lambda_{i} \right] \right\} = \sum_{i=1}^{m} \max \left\{0, \Re \left[ \lambda_{i} \left( A \right)\right] \right\}, 
\end{flalign}
which holds since the system is detectable (all the unstable modes of the system are observable) and thus the unstable poles of 
$L \left( s \right)$ correspond to the eigenvalues of $A$ with real parts larger than zero. 
It is also worth mentioning that when it is further assumed that $l=1$, \eqref{cBode} reduces to 
\begin{multline} \label{cBode2}
\frac{1}{2\pi}\int_{-\infty}^{\infty} \ln \left| \frac{1}{1 + L \left( \mathrm{j} \omega \right)} \right| \mathrm{d} \omega \\
=   - \lim_{s \to \infty} \frac{1}{2} s L \left( s \right) + \sum_{i=1}^{m} \max \left\{0, \Re \left[ \lambda_{i} \right] \right\}.
\end{multline}

On the other hand, it can be verified that the Bode integrals given in \eqref{cBode} and \eqref{cBode2} are implied by \eqref{cJensen} by letting $f \left( s \right) = \det \left[ I + L \left(  s  \right) \right]^{-1}$ and $f \left( s \right) = 1/ \left[ 1 + L \left(  s  \right) \right]$, respectively. In addition, note that one common instance for 
\begin{flalign}
& \lim_{s \to \infty} s \left\{ \ln \left| \det \left[ I + L \left(  s  \right) \right]^{-1} \right| \right\} = 0 
\end{flalign} 
to hold is that all the entries of $L \left(  s  \right)$ have at least two more poles than zeros, which is not satisfied in the case of Kalman--Bucy filtering as discussed in this paper.

\subsection{A More Explicit Expression of \eqref{cgeneral1}}

Using \eqref{cJensen2}, we can obtain a more explicit form for the formula in \eqref{cgeneral1}.

\begin{theorem} \label{alternative}
	Equality \eqref{cgeneral1} can be rewritten as
	\begin{flalign} \label{general1e}
	\mathrm{tr} \left( \Sigma_{\mathbf{z} - \overline{\mathbf{y}}} V^{-1} \right) 
	&= \mathrm{tr} \left( CPC^T V^{-1} \right) \nonumber \\
	& = \sum_{j}^{} \max \left\{0, \Re \left[ \varphi_{j} \right] \right\} - \sum_{k}^{} \max \left\{0, \Re \left[ \eta_{k}\right] \right\} \nonumber \\
	&\ \ \ \  + 2 \sum_{i=1}^{m} \max \left\{0, \Re \left[ \lambda_{i} \left( A \right)\right] \right\}, 
	\end{flalign}
	where $\varphi_{j}$ denote the zeros of $\det \left[ \Phi_{\mathbf{y}} \left( s \right)V^{-1} \right]$ and $\eta_{k}$ denote its poles. 
\end{theorem}


\begin{proof} By invoking \eqref{cJensen2}, it holds that
	\begin{flalign}
	&\frac{1}{2 \mathrm{\pi}}\int_{-\infty}^{\infty}{\ln \left[ \frac{\det \Phi_{\mathbf{y}} \left( \omega \right)}{\det V} \right] \mathrm{d}\omega } \nonumber \\
	&\ \ \ \  = \frac{1}{2 \mathrm{\pi}}\int_{-\infty}^{\infty}{\ln \det \left[ \Phi_{\mathbf{y}} \left( \omega \right) V^{-1} \right]\mathrm{d}\omega } \nonumber \\
	&\ \ \ \  = \lim_{s \to \infty} \frac{1}{2} s \left\{ \ln \det \left[ \Phi_{\mathbf{y}} \left( s \right) V^{-1} \right] \right\} + \sum_{j}^{} \max \left\{0, \Re \left[ \varphi_{j} \right] \right\} \nonumber \\
	&\ \ \ \ \ \ \ \  - \sum_{k}^{} \max \left\{0, \Re \left[ \eta_{k}\right] \right\}. \nonumber 
	\end{flalign}
	Since
	\begin{flalign}
	&\Phi_{\mathbf{y}} \left( s \right) V^{-1} 
	=  C \left( s I - A \right)^{-1} W \left( - s I - A \right)^{-T} C^T V^{-1} + I, \nonumber 
	\end{flalign}
	it follows from \cite{mustafa1990minimum} (Appendix~A) that
	\begin{flalign}
	&\lim_{s \to \infty} \frac{1}{2} s \left\{ \ln \det \left[ \Phi_{\mathbf{y}} \left( s \right) V^{-1} \right] \right\} = 0, \nonumber 
	\end{flalign}
	Noting also \eqref{cgeneral1}, \eqref{general1e} follows by simple substitution. 
\end{proof} 


Concerning the terms in \eqref{general1e}, it holds in general that
\begin{flalign}
& \sum_{k}^{} \max \left\{0, \Re \left[ \eta_{k}\right] \right\}
\leq 2 \sum_{i=1}^{m} \max \left\{0, \Re \left[ \lambda_{i} \left( A \right)\right] \right\}, 
\end{flalign}
since the unstable poles of $\det \left[ \Phi_{\mathbf{y}} \left( s \right)V^{-1} \right]$ must belong to the set of eigenvalues $A$ with real parts larger than zero, while not all such eigenvalues of $A$ may be present in the set of unstable poles of $ \det \left[ \Phi_{\mathbf{y}} \left( s \right)V^{-1} \right]$. 
Note that this is different from the case of $\det \left[ I + C \left(  s I - A \right)^{-1} K \right]$, as discussed in the proof of Theorem~\ref{cgeneral}; in such a case, since the system $\left( A, C \right)$ is detectable (all the unstable modes of the system are observable), the set of unstable poles of 
$\det \left[ I + C \left(  s I - A \right)^{-1} K \right]$ is exactly the same as the set of eigenvalues of $A$ with real parts larger than zero.

\subsection{Some Straightforward Corollaries}

Using the fact that 
\begin{flalign}
\underline{\lambda} \left( C^T V^{-1} C \right) \mathrm{tr} P
&\leq \mathrm{tr} \left( CPC^T V^{-1} \right) \nonumber \\
&= \mathrm{tr} \left( V^{-\frac{1}{2}} CPC^T V^{-\frac{1}{2}} \right) \nonumber \\
&\leq \overline{\lambda} \left( C^T V^{-1} C \right) \mathrm{tr} P, 
\end{flalign}
we could obtain lower and upper bounds on $\mathrm{tr} P$ based on \eqref{cgeneral1} as
\begin{flalign}
\mathrm{tr} P
&\geq \frac{1}{\overline{\lambda} \left( C^T V^{-1} C \right) } \Bigg\{ \frac{1}{2\pi}\int_{-\infty}^{\infty} \ln \left[ \frac{\det \Phi_{\mathbf{y}} \left( \omega \right) }{\det V} \right] \mathrm{d} \omega \nonumber\\ 
&\ \ \ \  + 2 \sum_{i=1}^{m} \max \left\{0, \Re \left[ \lambda_{i} \left( A \right)\right] \right\} \Bigg\}, 
\end{flalign}
and
\begin{flalign}
\mathrm{tr} P
&\leq \frac{1}{\underline{\lambda} \left( C^T V^{-1} C \right) } \Bigg\{ \frac{1}{2\pi}\int_{-\infty}^{\infty} \ln \left[ \frac{\det \Phi_{\mathbf{y}} \left( \omega \right) }{\det V} \right] \mathrm{d} \omega \nonumber\\ 
&\ \ \ \  + 2 \sum_{i=1}^{m} \max \left\{0, \Re \left[ \lambda_{i} \left( A \right)\right] \right\} \Bigg\}. 
\end{flalign}
One might also compare our bounds with those in, e.g., \cite{kwon1996bounds}, which, however, goes beyond the scope of this paper.

We now consider some special cases of Theorem~\ref{cgeneral}.

\begin{corollary} \label{cgeneral2}
	If $l=1$, and suppose that $\left\{ \mathbf{w} \left( t \right) \right\}$ and $\left\{ \mathbf{v} \left( t \right) \right\}$ are white Gaussian with covariance matrix $W$ and variance $\sigma_{\mathbf{v}}^2$, respectively, then
	\begin{flalign} \label{cgeneral23}
	&\sigma_{\mathbf{z} - \overline{\mathbf{y}}}^2
	=CPC^T  \nonumber \\
	&= \frac{\sigma_{\mathbf{v}}^2}{2\pi}\int_{-\infty}^{\infty} \ln \left[ \frac{S_{\mathbf{y}} \left( \omega \right)}{\sigma_{\mathbf{v}}^2} \right] \mathrm{d} \omega  
	+ 2 \sigma_{\mathbf{v}}^2 \sum_{i=1}^{m} \max \left\{0, \Re \left[ \lambda_{i} \left( A \right)\right] \right\}, 
	\end{flalign}
	where
	\begin{flalign}
	& S_{\mathbf{y}} \left( \omega \right) =  C \left( \mathrm{j} \omega I - A \right)^{-1} W \left( - \mathrm{j} \omega I - A \right)^{-T} C^T  + \sigma_{\mathbf{v}}^2. 
	\end{flalign}
\end{corollary}



\begin{corollary} 
	If $l=m=1$, and suppose that $\left\{ \mathbf{w} \left( t \right) \right\}$ and $\left\{ \mathbf{v} \left( t \right) \right\}$ are white Gaussian with variances $\sigma_{\mathbf{w}}^2$ and $\sigma_{\mathbf{v}}^2$, respectively, then
	\begin{flalign}
	P
	&= \frac{1}{C^2} \Bigg\{\frac{\sigma_{\mathbf{v}}^2}{2\pi}\int_{-\infty}^{\infty} \ln \left[ C^2 \left|  \frac{1}{\mathrm{j} \omega - A}  \right|^2 \frac{\sigma_{\mathbf{w}}^2}{\sigma_{\mathbf{v}}^2} + 1 \right] \mathrm{d} \omega \nonumber \\
	&\ \ \ \  + 2 \sigma_{\mathbf{v}}^2 \max \left\{0, \Re \left[ A \right] \right\}\Bigg\}. 
	\end{flalign}
\end{corollary}

\subsection{Relation to Existing Results}

We now present a list of comparisons with existing results in the related works. It will be seen that our result in \eqref{cgeneral1} is the most generic and reduces to all the listed ones, although they have adopted different approaches.

\begin{itemize}
	\item When $l=1$ and $A$ is stable, $\left\{ \mathbf{\mathbf{z}} \left( t \right) \right\}$ is stationary and 
	\begin{flalign}
	& C \left( \mathrm{j} \omega I - A \right)^{-1} W \left(  - \mathrm{j} \omega I - A \right)^{-T} C^T 
	\end{flalign} 
	denotes its power spectrum $S_{\mathbf{z}} \left( \omega \right)$. Hence, \eqref{cgeneral23} reduces to
	\begin{flalign} \label{cKS}
	&\sigma_{\mathbf{z} - \overline{\mathbf{y}}}^2 = CPC^T
	= \frac{\sigma_{\mathbf{v}}^2}{2\pi}\int_{-\infty}^{\infty} \ln \left[ \frac{S_{\mathbf{z}} \left( \omega \right) + \sigma_{\mathbf{v}}^2}{\sigma_{\mathbf{v}}^2} \right] \mathrm{d} \omega. 
	\end{flalign}
	This coincides with the result in \cite{van1968detection}.
	\item When $l=1$, $\sigma_{\mathbf{v}}^2=1$, and $W=BB^{T}$, where $B \in \mathbb{R}^{m \times 1}$, \eqref{cgeneral23} reduces to
	\begin{flalign}
	\sigma_{\mathbf{z} - \overline{\mathbf{y}}}^2 
	&= CPC^T  \nonumber \\
	&= \frac{1}{2\pi}\int_{-\infty}^{\infty} \ln \left[ \left|C \left( \mathrm{j} \omega I - A \right)^{-1} B\right|^2  + 1 \right] \mathrm{d} \omega  \nonumber \\
	&\ \ \ \  + 2 \sum_{i=1}^{m} \max \left\{0, \Re \left[ \lambda_{i} \left( A \right)\right] \right\}. 
	\end{flalign}
	This coincides with the conclusion in \cite{anderson1985use}.
	\item When $l=m=1$, $\sigma_{\mathbf{v}}^2 = 1$, and $\sigma_{\mathbf{w}}^2 = \varepsilon^2$, \eqref{cgeneral23} reduces to
	\begin{flalign}
	\sigma_{\mathbf{z} - \overline{\mathbf{y}}}^2 
	&= CPC^T  \nonumber \\
	&= \frac{1}{2\pi}\int_{-\infty}^{\infty} \ln \left[ \varepsilon^2 C^2 \left|  \frac{1}{\mathrm{j} \omega - A}  \right|^2  + 1 \right] \mathrm{d} \omega  \nonumber \\
	&\ \ \ \  + 2 \sum_{i=1}^{m} \max \left\{0, \Re \left[ A \right] \right\}. 
	\end{flalign}
	This coincides with the corresponding result in \cite{braslavsky1999limiting}.
	\item When $l=1$, $\sigma_{\mathbf{v}}^2 = 1$, and $W =0$, \eqref{cgeneral23} reduces to
	\begin{flalign}
	& \sigma_{\mathbf{z} - \overline{\mathbf{y}}}^2 
	= CPC^T
	= 2 \sum_{i=1}^{m} \max \left\{0, \Re \left[ \lambda_{i} \left( A \right)\right] \right\}. 
	\end{flalign}
	This also coincides with conclusion in \cite{braslavsky1999limiting}.
	\item When $V = I$, \eqref{cgeneral1} reduces to
	\begin{flalign}
	\mathrm{tr}	\Sigma_{\mathbf{z} - \overline{\mathbf{y}}} 
	&= \mathrm{tr} \left( CPC^T \right) \nonumber \\
	&=  \frac{1}{2\pi}\int_{-\infty}^{\infty} \ln \det \Phi_{\mathbf{y}} \left( \omega \right) \mathrm{d} \omega \nonumber \\
	&\ \ \ \  + 2 \sum_{i=1}^{m} \max \left\{0, \Re \left[ \lambda_{i} \left( A \right)\right] \right\}, 
	\end{flalign}
	where
	\begin{flalign} 
	& \Phi_{\mathbf{y}} \left( \omega \right) =  C \left( \mathrm{j} \omega I - A \right)^{-1} W \left( - \mathrm{j} \omega I - A \right)^{-T} C^T  + I. 
	\end{flalign}
	This coincides with the result in \cite{zhang2012mmse}.
\end{itemize}

Discussions on the dual problem in control can be found in, e.g., \cite{anderson1985use, freudenberg2005control}.

Indeed, equation~\eqref{cKS} is referred to as the Yovits--Jackson formula in \cite{kailath1974view}.
In a broad sense, the Yovits--Jackson formula can be viewed as
the continuous-time counterpart of the Kolmogorov--Szeg\"o formula \cite{ linearestimation, Pap:02, lindquist2015linear, vaidyanathan2007theory, chen2018role} in the discrete-time case. For the multiple-input multiple-output (MIMO) case, when $\left\{ \mathbf{y} \left( t \right) \right\}$ is stationary ($A$ is stable), \eqref{cgeneral1} reduces to
\begin{flalign}
\mathrm{tr} \left( \Sigma_{\mathbf{z} - \overline{\mathbf{y}}} V^{-1} \right) 
&= \mathrm{tr} \left( CPC^T V^{-1} \right) \nonumber \\
& =  \frac{1}{2\pi}\int_{-\infty}^{\infty} \ln \left[ \frac{\det \Phi_{\mathbf{y}} \left( \omega \right)}{\det V} \right] \mathrm{d} \omega \nonumber \\
&=  \frac{1}{2\pi}\int_{-\infty}^{\infty} \ln \left\{\frac{\det \left[ \Phi_{\mathbf{z}} \left( \omega \right) + V \right]}{\det V} \right\}  \mathrm{d} \omega,
\end{flalign}
which may be viewed as the as the continuous-time counterpart of the Wiener--Masani formula \cite{wiener1957prediction, lindquist2015linear, chen2018role}.  

As we pointed out in \cite{FangACC18}, our formula therein generalizes the Kolmogorov--Szeg\"o formula and the Wiener--Masani formula to the non-stationary case, and correspondingly in this paper, we have generalized the Yovits--Jackson formula to the MIMO case as well as the non-stationary case.

\section{Conclusion}

%

In this paper, we have shown that the trace of the division of the optimal output estimation error covariance over the noise covariance attained by the Kalman--Bucy filter can be explicitly expressed in terms of the plant dynamics and noise statistics in a frequency-domain integral characterization. We have also discussed the relation of our integral characterization to the Bode integral. Possible future research directions include analysis of other classes of algebraic Riccati equations. 

\section*{Appendix}

\subsection{Proof of Proposition~\ref{Jensen}} \label{Jensenproof}

It follows from the Poisson--Jensen formula for the half plane \cite{goodwin2001control} that for every point $s_{0} = \sigma_{0} + \mathrm{j} \omega_{0}$, $\sigma_{0} >0$, 
\begin{multline} \label{PJ}
\frac{1}{\pi}\int_{-\infty}^{\infty} \frac{\sigma_{0}}{\sigma_{0}^2 + \left( \omega - \omega_{0} \right)^2} \ln \left| f \left( \mathrm{j} \omega \right) \right| \mathrm{d} \omega \\ = \ln \left| f \left( s_{0} \right) \right| - \sum_{i}^{} \ln \left| \frac{s_{0} - z_{i}}{s_{0} + z_{i}^{*}} \right|, 
\end{multline}
where $z_{i}$ denote the nonminimum-phase zeros of $f \left( s \right)$ and $z_{i}^{*}$ denote the complex conjugate of $z_{i}$. Hence, taking $\omega_{0} = 0$, we have
\begin{multline}
\frac{1}{\pi}\int_{-\infty}^{\infty} \frac{\sigma_{0}}{\sigma_{0}^2 + \omega^2} \ln \left| f \left( \mathrm{j} \omega \right) \right| \mathrm{d} \omega \\ = \ln \left| f \left( \sigma_{0} \right) \right| - \sum_{i}^{} \ln \left| \frac{\sigma_{0} - z_{i}}{\sigma_{0} + z_{i}^{*}} \right|.
\end{multline}
As such,
\begin{flalign} \label{doyle}
&\lim_{\sigma_{0} \to \infty} \sigma_{0} \left[ \frac{1}{2 \pi}\int_{-\infty}^{\infty} \frac{\sigma_{0}}{\sigma_{0}^2 + \omega^2}  \ln \left| f \left( \mathrm{j} \omega \right) \right| \mathrm{d} \omega \right] \nonumber \\
&\ \ \ \  = \lim_{\sigma_{0} \to \infty} \left[ \frac{1}{2} \sigma_{0} \ln \left| f \left( \sigma_{0} \right) \right| \right] \nonumber \\
&\ \ \ \ \ \ \ \  - \lim_{\sigma_{0}\to \infty} \left(  \frac{1}{2} \sigma_{0} \sum_{i}^{} \ln \left| \frac{\sigma_{0} - z_{i}}{\sigma_{0} + z_{i}^{*}} \right| \right). 
\end{flalign}
It is known \cite{doyle1990feedback} that
\begin{flalign}
&\lim_{\sigma_{0} \to \infty} \sigma_{0} \left[ \frac{1}{2 \pi}\int_{-\infty}^{\infty} \frac{\sigma_{0}}{\sigma_{0}^2 + \omega^2} \ln \left| f \left( \mathrm{j} \omega \right) \right|  \mathrm{d} \omega \right] \nonumber \\
&\ \ \ \  =\lim_{\sigma_{0} \to \infty} \left[ \frac{1}{2 \pi}\int_{-\infty}^{\infty} \frac{\sigma_{0}^2}{\sigma_{0}^2 + \omega^2}  \ln \left| f \left( \mathrm{j} \omega \right) \right| \mathrm{d} \omega \right] \nonumber \\
&\ \ \ \  = \frac{1}{2 \pi}\int_{-\infty}^{\infty} \ln \left| f \left( \mathrm{j} \omega \right) \right| \mathrm{d} \omega, \nonumber 
\end{flalign}
and
\begin{flalign}
- \lim_{\sigma_{0}\to \infty} \left( \frac{1}{2} \sigma_{0} \sum_{i}^{} \ln \left| \frac{\sigma_{0} - z_{i}}{\sigma_{0} + z_{i}^{*}} \right| \right) 
&= \sum_{i}^{} \Re \left[ z_{i} \right] \nonumber \\
&= \sum_{i=1}^{m} \max \left\{0, \Re \left[ \varphi_{i} \right] \right\}. \nonumber 
\end{flalign}
We finally consider the remaining term in \eqref{doyle}:
\begin{flalign} \label{doyle1}
&\lim_{\sigma_{0} \to \infty} \left[ \frac{1}{2} \sigma_{0} \ln \left| f \left( \sigma_{0} \right) \right| \right]. 
\end{flalign}
Denote 
$ p \left( s \right) = p_m s^m + p_{m-1} s^{m-1} + \cdots + p_1 s + p_0$, and 
$ q \left( s \right) = q_m s^m + q_{m-1} s^{m-1} + \cdots + q_1 s + q_0$.
Since $\lim_{s \to \infty} f \left( s \right) = 1$, we have $p_m = q_m$. In addition, as $s \to \infty$, we have
\begin{flalign}
f \left( s \right) 
&= \frac{p \left( s \right)}{q \left( s \right)} = \frac{p_m s^m + p_{m-1} s^{m-1} + \cdots + p_1 s + p_0}{q_m s^m + q_{m-1} s^{m-1} + \cdots + q_1 s + q_0} \nonumber \\
&= \frac{1 + \frac{p_{m-1}}{p_m} \frac{1}{s} + \frac{p_{m-2}}{p_m} \frac{1}{s}+ \cdots + \frac{p_{1}}{p_m} \frac{1}{s^{m-1}} + \frac{p_{0} }{p_m} \frac{1}{s^{m}}}{1 + \frac{q_{m-1}}{q_m} \frac{1}{s} + \frac{q_{m-2}}{q_m} \frac{1}{s} + \cdots + \frac{q_{1}}{q_m} \frac{1}{s^{m-1}} + \frac{q_{0} }{q_m} \frac{1}{s^{m}}} \nonumber \\
&= 1 +  \left( \frac{p_{m-1}}{p_m} - \frac{q_{m-1}}{q_m} \right) \frac{1}{s} + O \left(  \frac{1}{s^2} \right). \nonumber 
\end{flalign}
Hence, 
\begin{flalign}
&\lim_{\sigma_{0} \to \infty} \left[ \frac{1}{2} \sigma_{0} \ln \left| f \left( \sigma_{0} \right) \right| \right] \nonumber \\
&= \lim_{\sigma_{0} \to \infty} \left\{ \frac{1}{2} \sigma_{0} \ln \left[ 1 +  \left( \frac{p_{m-1}}{p_m} - \frac{q_{m-1}}{q_m} \right) \frac{1}{\sigma_{0}} + O \left(  \frac{1}{\sigma_{0}^2} \right)\right] \right\} \nonumber \\
&= \lim_{\sigma_{0} \to \infty} \left\{ \frac{1}{2} \left[ \left( \frac{p_{m-1}}{p_m} - \frac{q_{m-1}}{q_m} \right) + O \left(  \frac{1}{\sigma_{0}} \right)\right] \right\} \nonumber \\
&= \frac{1}{2} \left( \frac{p_{m-1}}{p_m} - \frac{q_{m-1}}{q_m} \right). \nonumber 
\end{flalign}
It thus follows that the limit 
\eqref{doyle1}
always exists; it is always finite, although it can be $0$ when $p_{m-1} = q_{m-1}$. This completes the proof.

\subsection{Proof of Proposition~~\ref{Jensen2}} \label{Jensen2proof}

Rewrite
\begin{flalign}
& f \left( s \right) = \frac{p \left( s \right)}{q \left( s \right)} = \frac{p \left( s \right)}{l \left( s \right)} \frac{l \left( s \right)}{q \left( s \right)}, \nonumber 
\end{flalign}
where $l \left( s \right)$ is chosen as a polynomial of $s$ satisfying the following conditions:
\begin{itemize}
	\item $l \left( s \right)$ is with the same order as those of $p \left( s \right)$ and $q \left( s \right)$;
	\item all the zeros of $l \left( s \right)$ are nonminimum phase;
	\item there is no pole-zero cancellation between $l \left( s \right)$ and $p \left( s \right)$ or between $l \left( s \right)$ and $q \left( s \right)$; and
	\item the following relation holds:
	\begin{flalign}
	&\lim_{s \to \infty} \frac{p \left( s \right)}{l \left( s \right)} = \lim_{s \to \infty} \frac{l \left( s \right)}{q \left( s \right)} = 1. \nonumber 
	\end{flalign}
\end{itemize}
Therefore, the zeros of $f \left( s \right)$ are given by the zeros of $p \left( s \right)/l \left( s \right)$, while the poles of $f \left( s \right)$ are given by the zeros of $q \left( s \right)/l \left( s \right)$. On the other hand, the poles of both $p \left( s \right)/l \left( s \right)$ and $q \left( s \right)/l \left( s \right)$ are stable. Then, using Proposition~\ref{Jensen}, we arrive at
\begin{flalign}
&\frac{1}{2\pi}\int_{-\infty}^{\infty} \ln \left| f \left( \mathrm{j} \omega \right) \right| \mathrm{d} \omega \nonumber \\
&\ \ \ \  = \frac{1}{2\pi}\int_{-\infty}^{\infty} \ln \left| \frac{p \left( \mathrm{j} \omega \right)}{l \left( \mathrm{j} \omega \right)}  \right| \mathrm{d} \omega -  \frac{1}{2\pi}\int_{-\infty}^{\infty} \ln \left| \frac{q \left( \mathrm{j} \omega \right)}{l \left( \mathrm{j} \omega \right)}  \right| \mathrm{d} \omega \nonumber \\
&\ \ \ \  =  \lim_{s \to \infty} \frac{1}{2} s \ln \left| \frac{p \left( s \right)}{l \left( s \right)} \right| + \sum_{i}^{} \max \left\{0, \Re \left[ \varphi_{i} \right] \right\}
\nonumber \\
&\ \ \ \ \ \ \ \  - \lim_{s \to \infty} \frac{1}{2} s \ln \left| \frac{q \left( s \right)}{l \left( s \right)} \right|  - \sum_{j}^{} \max \left\{0, \Re \left[ \eta_{j} \right] \right\}
\nonumber \\
&\ \ \ \  =  \lim_{s \to \infty} \frac{1}{2} s \ln \left| \frac{p \left( s \right)}{q \left( s \right)}  \right| + \sum_{i}^{} \max \left\{0, \Re \left[ \varphi_{i} \right] \right\}
\nonumber \\
&\ \ \ \ \ \ \ \  - \sum_{j}^{} \max \left\{0, \Re \left[ \eta_{j} \right] \right\} \nonumber \\
&\ \ \ \  =  \lim_{s \to \infty} \frac{1}{2} s \ln \left| f \left( s \right)  \right| + \sum_{i}^{} \max \left\{0, \Re \left[ \varphi_{i} \right] \right\}
\nonumber \\
&\ \ \ \ \ \ \ \  - \sum_{j}^{} \max \left\{0, \Re \left[ \eta_{j} \right] \right\}. \nonumber
\end{flalign}
This completes the proof.

\subsection{Proof of Theorem~\ref{cgeneral}} \label{cgeneralproof}

For all $s$, it holds that
\begin{flalign}
& A P + P A^T 
= \left( A - sI \right)P + P \left( A + sI \right)^T. \nonumber 
\end{flalign}
Hence, \eqref{cARE} can be rewritten as
\begin{flalign} 
& \left( A - sI \right)P + P \left( A + sI \right)^T + W - PC^{T}V^{{-1}} CP = 0, \nonumber 
\end{flalign}
or equivalently,
\begin{flalign} 
& \left( sI - A \right)P + P \left( - sI - A \right)^T = W - PC^{T}V^{{-1}} CP. \nonumber 
\end{flalign}
Thus,
\begin{flalign}
& C \left( s I - A \right)^{-1} \left( sI - A \right)P \left( -s I - A \right)^{-T} C^T \nonumber \\
&\ \ \ \  + C \left( s I - A \right)^{-1} P \left( - sI - A \right)^T \left( -s I - A \right)^{-T} C^T  \nonumber \\
&\ \ \ \  = C \left( s I - A \right)^{-1} W \left( -s I - A \right)^{-T} C^T \nonumber \\
&\ \ \ \ \ \ \ \  - C \left( s I - A \right)^{-1} P C^T V^{-1} C P \left( -s I - A \right)^{-T} C^T, \nonumber 
\end{flalign}
that is,
\begin{flalign}
& C P \left( -s I - A \right)^{-T} C^T + C \left( s I - A \right)^{-1} P  C^T  \nonumber \\
&\ \ \ \  = C \left( s I - A \right)^{-1} W \left( -s I - A \right)^{-T} C^T \nonumber \\
&\ \ \ \ \ \ \ \  - C \left( s I - A \right)^{-1} P C^T V^{-1} C P \left( -s I - A \right)^{-T} C^T. \nonumber 
\end{flalign}
In addition, noting \eqref{cAREK}, we have
\begin{flalign}
& C \left( s I - A \right)^{-1} P  C^T + C P \left( -s I - A \right)^{-T} C^T \nonumber \\
&\ \ \ \  + C \left( s I - A \right)^{-1} P C^T V^{-1} C P \left( -s I - A \right)^{-T} C^T  \nonumber \\
&= C \left( s I - A \right)^{-1} P  C^T V^{-1} V + V V^{-1} C P \left( -s I - A \right)^{-T} C^T \nonumber \\
&\ \ \ \  + C \left( s I - A \right)^{-1} P C^T V^{-1} V V^{-1} C P \left( -s I - A \right)^{-T} C^T  \nonumber \\
&= C \left( s I - A \right)^{-1} K V + V K^{T} \left( -s I - A \right)^{-T} C^T \nonumber \\
&\ \ \ \  + C \left( s I - A \right)^{-1} K V K^T \left( -s I - A \right)^{-T} C^T  \nonumber \\
&= C \left( s I - A \right)^{-1} W \left( -s I - A \right)^{-T} C^T. \nonumber  
\end{flalign}
Thus,
\begin{flalign}
& V + C \left( s I - A \right)^{-1} K V + V K^{T} \left( -s I - A \right)^{-T} C^T \nonumber \\
&\ \ \ \  + C \left( s I - A \right)^{-1} K V K^T \left( -s I - A \right)^{-T} C^T  \nonumber \\
&\ \ \ \  = C \left( s I - A \right)^{-1} W \left( -s I - A \right)^{-T} C^T + V, \nonumber 
\end{flalign}
and
\begin{flalign} \label{cfactorization}
& \left[ I + C \left( sI - A \right)^{-1} K \right] V \left[ I + C \left( -s I - A \right)^{-1} K \right]^T \nonumber \\
&\ \ \ \  = C \left( s I - A \right)^{-1} W \left( -s I - A \right)^{-T} C^T + V. 
\end{flalign}
As such,
\begin{flalign}
& \det \left\{ \left[ I + C \left( sI - A \right)^{-1} K \right] V \left[ I + C \left( -s I - A \right)^{-1} K \right]^T \right\} \nonumber \\
&\ \ \ \  = \left| \det  \left[ I + C \left( sI - A \right)^{-1} K \right] \right|^2 \det V \nonumber \\
&\ \ \ \  = \det \left[ C \left( s I - A \right)^{-1} W \left( -s I - A \right)^{-T} C^T + V \right], \nonumber 
\end{flalign}
and 
\begin{flalign}
& \ln \left| \det  \left[ I + C \left( sI - A \right)^{-1} K \right] \right|^2 + \ln \det V \nonumber \\
&\ \ \ \  = \ln \det \left[ C \left( s I - A \right)^{-1} W \left( -s I - A \right)^{-T} C^T + V \right]. \nonumber 
\end{flalign}
Hence,
\begin{flalign}
&\frac{1}{2\pi}\int_{- \infty}^{\infty} \ln \left| \det \left[ I + C \left( \mathrm{j} \omega I - A \right)^{-1} K \right] \right|^2 \mathrm{d} \omega \nonumber \\
&\ \ \ \  + \frac{1}{2\pi}\int_{- \infty}^{\infty} \ln \det V \mathrm{d} \omega \nonumber \\
&\ \ \ \  = \frac{1}{2\pi}\int_{- \infty}^{\infty} \ln \det \Phi_{\mathbf{y}} \left( \omega \right) \mathrm{d} \omega,  \nonumber
\end{flalign}
where $\Phi_{\mathbf{y}} \left( \omega \right)$ is given by \eqref{cnotation}. Consequently, 
\begin{flalign}
&- \frac{1}{2\pi}\int_{- \infty}^{\infty} \ln \left| \det \left[ I + C \left( \mathrm{j} \omega I - A \right)^{-1} K \right]^{-1} \right|^2 \mathrm{d} \omega \nonumber \\
&\ \ \ \  = \frac{1}{2\pi}\int_{- \infty}^{\infty} \ln \left[ \frac{\det \Phi_{\mathbf{y}} \left( \omega \right)}{\det V} \right] \mathrm{d} \omega. \nonumber 
\end{flalign}
Note that
\begin{flalign}
\left[ I + C \left(  s I - A \right)^{-1} K \right]^{-1} =  I - C \left[ s I - \left(A - KC\right) \right] ^{-1} K. \nonumber &
\end{flalign}
Hence, the poles of 
\begin{flalign} \label{detL}
&\det \left[ I + C \left(  s I - A \right)^{-1} K \right]^{-1} 
\end{flalign} 
are indeed the eigenvalues of $A - KC$; since the Kalman--Bucy filter is asymptotically stable, all the poles of \eqref{detL} are stable. On the other hand, since the system is detectable, all the unstable modes of the system are observable, and thus the nonminimum-phase zeros of \eqref{detL} (the unstable poles of 
$\det \left[ I + C \left( s I - A \right)^{-1} K \right]$) correspond to the eigenvalues of $A$ with real parts larger than zero. As such, by letting 
\begin{flalign}
&f \left( s \right) = \det \left[ I + C \left(  s I - A \right)^{-1} K \right]^{-1}, \nonumber 
\end{flalign}
it follows from \eqref{cJensen} that
\begin{flalign} \label{cinvBode_dynamic}
&\frac{1}{2\pi}\int_{-\infty}^{\infty} \ln \left| \det \left[ I + C \left( \mathrm{j} \omega I - A \right)^{-1} K \right]^{-1} \right| \mathrm{d} \omega \nonumber \\
&\ \ \ \  =  \lim_{s \to \infty} \frac{1}{2} s \left\{ \ln \left| \det \left[ I + C \left( s I - A \right)^{-1} K \right]^{-1} \right| \right\} \nonumber \\
&\ \ \ \ \ \ \ \  + \sum_{i=1}^{m} \max \left\{0, \Re \left[ \varphi_{i} \right] \right\} \nonumber \\
&\ \ \ \  = - \frac{1}{2} \mathrm{tr} \left( C P C^{T} V^{-1} \right) + \sum_{i=1}^{m} \max \left\{0, \Re \left[ \lambda_{i} \left( A \right)\right] \right\}, 
\end{flalign}
where $\varphi_{i}$ denote the zeros of $\det \left[ I + C \left( s I - A \right)^{-1} K \right]^{-1}$. Herein, we have also used the fact that (see, e.g., Appendix~A of \cite{mustafa1990minimum})
\begin{flalign} 
&\lim_{s \to \infty} \frac{1}{2} s \left\{ \ln \left| \det \left[ I + C \left( s I - A \right)^{-1} K \right]^{-1} \right| \right\} \nonumber \\
&\ \ \ \  =- \lim_{s \to \infty} \frac{1}{2} s \left\{ \mathrm{tr} \left[ C \left( s I - A \right)^{-1} K \right] \right\} \nonumber \\
&\ \ \ \  =- \frac{1}{2} \mathrm{tr} \left( C K \right) \nonumber \\
&\ \ \ \  = - \frac{1}{2} \mathrm{tr} \left( C P C^{T} V^{-1} \right). \nonumber 
\end{flalign}
As such,
\begin{flalign}
&\mathrm{tr} \left( \Sigma_{\mathbf{z} - \overline{\mathbf{y}}} \right)
= \mathrm{tr} \left( CPC^T V^{-1} \right) \nonumber \\
&\ \ \ \ = - \frac{1}{2\pi}\int_{-\infty}^{\infty} \ln \left| \det \left[ I + C \left( \mathrm{j} \omega I - A \right)^{-1} K \right]^{-1} \right|^2 \mathrm{d} \omega \nonumber \\
&\ \ \ \ \ \ \ \  + 2 \sum_{i=1}^{m} \max \left\{0, \Re \left[ \lambda_{i} \left( A \right)\right] \right\} \nonumber \\
&\ \ \ \ =  \frac{1}{2\pi}\int_{-\infty}^{\infty} \ln \left[ \frac{\det \Phi_{\mathbf{y}} \left( \omega \right) }{\det V} \right] \mathrm{d} \omega \nonumber \\
&\ \ \ \ \ \ \ \  + 2 \sum_{i=1}^{m} \max \left\{0, \Re \left[ \lambda_{i} \left( A \right)\right] \right\}. \nonumber 
\end{flalign}
This completes the proof.

\bibliographystyle{IEEEtran}
\bibliography{uctest}

\end{document}